\documentclass[11pt]{article}
\usepackage{fullpage}
\usepackage{natbib}

\usepackage{amsmath}
\usepackage{amsthm}
\usepackage{graphicx}
\usepackage{epsfig}
\usepackage{subfigure}
\usepackage[linktocpage=true,linkcolor=magenta,colorlinks,citecolor=blue]{hyperref}

\theoremstyle{plain}
\newtheorem{theorem}	 			{Theorem}[section]

\theoremstyle{definition}
 
\theoremstyle{remark} 

\usepackage{xspace}
\usepackage{amssymb}
\usepackage{nicefrac}
\usepackage{xcolor}

\title{Random-Order Models\footnote{Preprint of Chapter 11 in  \emph{Beyond the Worst-Case Analysis of Algorithms}, Cambridge University Press, 2020.}}
\author{Anupam Gupta\thanks{Computer Science Department, Carnegie Mellon University.  Email: \texttt{anupamg@cs.cmu.edu}. Supported in part by NSF award CCF-1907820.} \and Sahil Singla\thanks{Princeton University and Institute for Advanced Study.   Email: \texttt{singla@cs.princeton.edu}. Supported in part by the Schmidt Foundation.}}

\begin{document}

\maketitle


\newcommand{\junk}[1]{}
\newcommand{\ignore}[1]{}
\newcommand{\R}[0]{{\ensuremath{\mathbb{R}}}}
\newcommand{\N}[0]{{\ensuremath{\mathbb{N}}}}
\newcommand{\Z}[0]{{\ensuremath{\mathbb{Z}}}}
\newcommand{\E}[0]{{\ensuremath{\mathbb{E}}}}

\def\floor#1{\lfloor #1 \rfloor}
\def\ceil#1{\lceil #1 \rceil}
\def\seq#1{\langle #1 \rangle}
\def\set#1{\{ #1 \}}
\def\abs#1{\mathopen| #1 \mathclose|}   
\def\norm#1{\mathopen\| #1 \mathclose\|}
\def\cin{h} 

\newcommand{\poly}{\operatorname{poly}}
\newcommand{\argmin}{\operatorname{argmin}}
\newcommand{\polylog}{\operatorname{polylog}}

\newcommand{\sse}{\subseteq}

\newcommand{\calM}{{\mathcal{M}}}
\newcommand{\calF}{{\mathcal{F}}}
\newcommand{\Fstar}{{\mathscr{F}^\star}}
\newcommand{\bB}{{\mathbf{B}}}
\newcommand{\T}{{\mathcal{T}}}
\newcommand{\Tstar}{{T^\star}}

\newcommand{\e}{\varepsilon}
\newcommand{\eps}{\varepsilon}
\newcommand{\ts}{\textstyle}
\renewcommand{\sp}{{\hspace*{0.1 in}}}
\newcommand{\symdif}{\triangle}

\newcommand{\bs}[1]{\boldsymbol{#1}}
\newcommand{\pr}[1]{{\rm Pr} \left[ #1 \right]}
\newcommand{\ex}[1]{{\rm E} \left[ #1 \right]}

\newcommand{\Opt}{\ensuremath{\mathsf{opt}\xspace}}
\newcommand{\MST}{\ensuremath{\mathsf{mst}\xspace}}
\newcommand{\OPT}{\ensuremath{\mathsf{opt}\xspace}}

\newcommand{\cost}{\ensuremath{\mathsf{cost}\xspace}}

\renewcommand{\theequation}{\thesection.\arabic{equation}}
\renewcommand{\thefigure}{\thesection.\arabic{figure}}

\newcounter{note}[section]
\renewcommand{\thenote}{\thesection.\arabic{note}}
\newcommand{\agnote}[1]{\refstepcounter{note}$\ll${\bf Anupam~\thenote:}
  {\sf \color{blue} #1}$\gg$\marginpar{\tiny\bf AG~\thenote}}
\newcommand{\ssnote}[1]{\refstepcounter{note}$\ll${\bf Sahil~\thenote:}
  {\sf \color{green} #1}$\gg$\marginpar{\tiny\bf AK~\thenote}}
\newcommand{\alert}[1]{{\color{red}#1}}
\newcommand{\blew}[1]{{\color{blue}#1}}

\newcommand{\qedsymb}{\hfill{\rule{2mm}{2mm}}}
\renewenvironment{proof}{\begin{trivlist} \item[\hspace{\labelsep}{\bf
\noindent Proof.\/}] }{\qedsymb\end{trivlist}}%

\newenvironment{analysis}{\begin{trivlist} \item[\hspace{\labelsep}{\bf
\noindent Analysis.\/}] }{\qedsymb\end{trivlist}}%

\newenvironment{algo}{\vspace{0.07in} \noindent \begin{minipage}	{\textwidth}
\hrule\vspace{0.01in}\hrule \vspace{0.05in}}{\vspace{-0.07in}\hrule\vspace{0.01in}\hrule \vspace{0.07in}
\end{minipage}}%

\newcounter{myLISTctr}
\newcommand{\initOneLiners}{%
    \setlength{\itemsep}{0pt}
    \setlength{\parsep }{0pt}
    \setlength{\topsep }{0pt}
}
\newenvironment{OneLiners}[1][\ensuremath{\bullet}]
    {\begin{list}
        {#1}
        {\initOneLiners}}
    {\end{list}}

\newcommand{\squishlist}{
 \begin{list}{$\bullet$}
  { \setlength{\itemsep}{0pt}
     \setlength{\parsep}{3pt}
     \setlength{\topsep}{3pt}
     \setlength{\partopsep}{0pt}
     \setlength{\leftmargin}{1.5em}
     \setlength{\labelwidth}{1em}
     \setlength{\labelsep}{0.5em} } }

\newcommand{\squishend}{
  \end{list}  }

\newcommand{\gr}{\nabla}
\newcommand{\ip}[1]{\langle #1 \rangle}
\newcommand{\ba}{\mathbf{a}}

\newcommand{\saveforlong}[1]{}
\newcommand{\onlyforshort}[1]{#1}


\begin{abstract}
\medskip
  This chapter introduces the {\em random-order model} in online
  algorithms. In this model, the input is chosen by an adversary, then
  randomly permuted before being presented to the algorithm.  This
  reshuffling often  weakens the power of the adversary and
  allows for improved algorithmic guarantees. We show such improvements
  for two broad classes of problems: packing problems where we must pick
  a constrained set of items to maximize total value, and covering
  problems where we must satisfy given requirements at minimum total
  cost. We also discuss how random-order model relates to other stochastic models
  used for non-worst-case competitive analysis.
  \qedsymb
\end{abstract}

\medskip
\medskip

\tableofcontents

\newpage



\section{Motivation: Picking a Large Element}
\label{sec:intro}

Suppose we want to pick the maximum of a collection of $n$ numbers.
 At the beginning, we know this cardinality $n$, but nothing
about the range of numbers to arrive. We are then presented distinct
non-negative real numbers $v_1, v_2, \ldots, v_n$ one by one; upon
seeing a number $v_i$, we must either immediately pick it, or 
discard it forever. We can pick at most one number. The goal is to
maximize the expected value of the number we pick, where the expectation
is over any randomness in our algorithm.  We want this expected value to
be close to the maximum value
$v_{\max} := \max_{i \in \{1,2,\ldots, n\}} v_i$. Formally, we want to
minimize the \emph{competitive ratio}, which is defined as the ratio of
$v_{\max}$ to our expected value.  Note that this maximum $v_{\max}$ is independent
of the order in which the elements are presented, and is unknown to our
algorithm until all the numbers have been revealed.

If we use a deterministic algorithm,  our value can be arbitrarily
smaller than $v_{\max}$, even for $n=2$. Say the first number $v_1 = 1$. If
our deterministic algorithm picks $v_1$, the adversary can present
$v_2 = M \gg 1$; if it does not,
the adversary can
present $v_2 = \nicefrac{1}{M} \ll 1$. Either way, 
the adversary can make the competitive ratio as bad as it wants by making $M$ large.

Using a randomized algorithm helps only a little: a na\"{\i}ve
randomized strategy is to select a uniformly random position
$i \in \{1,\ldots, n\}$ up-front and pick the $i^{th}$ number
$v_i$. Since we pick each number with probability $1/n$, the expected
value is $\sum_i v_i/n \geq v_{\max}/n$. This turns out to be the best
we can do, as long the input sequence is controlled by an adversary and
the maximum value is much larger than the others. Indeed, one strategy
for the adversary is to choose a uniformly random index $j$, and present
the request sequence $1, M, M^2, \ldots, M^j, 0, 0, \ldots, 0$---a
rapidly ascending chain of $j$ numbers followed by worthless numbers. If
$M$ is very large, any good algorithm must pick the last number in the
ascending chain upon seeing it. But this is tantamount to guessing $j$,
and random guessing is the best an algorithm can do. (This intuition
can be made formal using Yao's minimax lemma.)





These bad examples show that the problem is hard for two reasons: the
first reason being the large range of the numbers involved, and the
second being the adversary's ability to carefully design these difficult
sequences.  Consider the following way to mitigate the latter effect:
what if the adversary chooses the $n$ numbers, but then the numbers are
shuffled and presented to the algorithm in a uniformly random order?
This random-order version of the problem above is commonly known as the
\emph{secretary problem}: the goal is to hire the best secretary (or at
least a fairly good one) if the candidates for the job appear in a
random order.

Somewhat surprisingly, randomly shuffling the numbers changes the
complexity of the problem drastically. Here is the elegant $50\%$-algorithm:

\begin{algo}
  \begin{enumerate}
  \item Reject the first $n/2$ numbers, and then 
  \item Pick the first number
      after that which is bigger than all the previous numbers  (if any).
  \end{enumerate}
\end{algo}
\begin{theorem} \label{thm:50Percent}
  The 50\%-algorithm gets an expected value of at least $v_{\max}/4$.
\end{theorem}
\begin{proof}
  Assume for simplicity all numbers are distinct.
  The algorithm definitely picks $v_{\max}$ if the highest number is
  in the second half of the random order (which happens with probability
  $\nicefrac12$), and also the second-highest number is in the first half
  (which, conditioned on the first event, happens with probability at
  least $\nicefrac12$, the two events being positively
  correlated). Hence, we get an expected value of at
  least $v_{\max}/4$. (We get a stronger guarantee: we pick the highest
  number $v_{\max}$ itself with probability at least $\nicefrac14$, but we will not explore this
  expected-value-versus-probability direction any further.)
\end{proof}


\subsection{The Model and a Discussion}

The secretary problem, with the lower bounds in the worst-case setting
and an elegant algorithm for the random-order model, highlights the fact
that sequential decision-making problems are often hard in the
worst-case not merely because the underlying \emph{set} of requests is
hard, but also because these requests are carefully woven into a
difficult-to-solve \emph{sequence}. In many situations where
there is no adversary, it may be reasonable to assume that the ordering
of the requests is benign, which leads us to the random-order model.
Indeed, one can view this as a \emph{semi-random} model from
Chapter~9, where
the input is first chosen by an adversary and then randomly perturbed
before being given to the algorithm.

Let us review the \emph{competitive analysis} model for worst-case
analysis of online
algorithms (also discussed in Chapter~24). Here, the adversary chooses a sequence
of requests and present them to the algorithm one by
one. The algorithm must take actions to serve a request before seeing
the next request, and it cannot change past decisions. The actions have
rewards, say, and the \emph{competitive ratio} is the optimal reward for
the sequence (in hindsight) divided by the algorithm's reward. (For
problems where we seek to minimize costs instead of maximize rewards,
the competitive ratio is the algorithm's cost divided by the optimal
cost.) Since the algorithm can never out-perform the optimal choice, the
competitive ratio is always at least $1$.

Now given any online problem, the \emph{random-order model} (henceforth the
\emph{RO model}) considers the setting where the adversary first chooses
a \emph{set} $S$ of requests (and not a sequence). The elements of this
set are then
presented to the algorithm in a uniformly random order. Formally, given
a set $S = \{r_1, r_2, \ldots, r_n\}$ of $n = |S|$ requests, we imagine
nature drawing a uniformly random permutation $\pi$ of
$\{1,\ldots, n\}$, and then defining the input sequence 
to be
$ r_{\pi(1)},  r_{\pi(2)}, \cdots, 
r_{\pi(n)}$. As before, the online algorithm sees these requests one by
one, and has to perform its irrevocable actions for $r_{\pi(i)}$ before
seeing $r_{\pi(i+1)}$.
The length $n$ of the input sequence may also be
revealed to the algorithm at the beginning, depending on the problem.
The competitive ratio (for maximization problems) is defined as the
ratio between the optimum value for $S$ and the expected value of
the algorithm, where the expectation is now taken over both the randomness
of the reshuffle $\pi$ and that of the algorithm. (Again, we use the convention
that the competitive ratio is at least one, and hence have to flip the
ratio for minimization problems.) 

A strength of the RO model is its simplicity, and that it captures other
commonly considered stochastic
input models.  Indeed, since the RO model does not assume the algorithm
has any prior knowledge of the underlying set of requests (except
perhaps the cardinality $n$), it captures situations where the input
sequence consists of independent and identically distributed (i.i.d.)
random draws from some fixed and unknown distribution. Reasoning
about the RO model avoids over-fitting the algorithm to any particular
properties of the distribution, and makes the algorithms more general
and robust by design. 

Another motivation for the RO model is
aesthetic and pragmatic: the simplicity of the model makes it a good
starting point for designing algorithms. If we want
to develop an online algorithm (or even an offline one) for some
algorithmic task, a good step is to first solve it in the RO model, and
then extend the result to the worst-case setting. This can be useful
either way: in the best case, we may succeed in getting an algorithm for  the worst-case setting
using the insights developed in the RO model. Else the extension may
be difficult, but still we know a good algorithm under the (mild?)
assumption of random-order arrivals.

Of course, the assumption of uniform random orderings may be
unreasonable in some settings, especially if the algorithm performs
poorly when the random-order assumption is violated. There have been
attempts to refine the model to require less randomness from the input
stream, while still getting better-than-worst-case performance. We
discuss some of these in \S\ref{sec:redRandom}, but much
remains to be done.


\subsection{Roadmap}

In \S\ref{sec:dynkin} we discuss an optimal algorithm for the
secretary problem. In \S\ref{sec:sec-ext} we give algorithms to choose
multiple items instead of just a single one, and other maximization
packing problems. In \S\ref{sec:comb} we discuss minimization problems.
In \S\ref{sec:others} we present some specializations and extensions
of the RO model.

\section{The Secretary Problem}
\label{sec:dynkin}

We saw the
$50\%$ algorithm 
based on the  idea of using the first half of the random order
sequence to compute a threshold that weeds out ``low'' values. 
This idea of choosing a good threshold will be a
recurring one in this chapter. 
The choice of waiting for half of the
sequence was for simplicity: a right choice is to wait for
$\nicefrac1e\approx 37\%$ fraction, which gives us the $37\%$-algorithm:

\begin{algo}
  \begin{enumerate}
  \item Reject the first $n/e$ numbers, and then
  \item Pick the first number
      after that (if any) which is bigger than all the previous numbers.
  \end{enumerate}
\end{algo}


\saveforlong{We will give two
proofs that this change increases the probability of picking the largest number
to $37\%$.  The first proof uses combinatorial arguments, whereas the
second is based on linear-programming.
}

(Although $n/e$ is not an integer,  rounding it to the nearest integer does
not  impact the guarantees substantively.)
Call a number a \emph{prefix-maximum} 
if it is the largest among the numbers revealed before it. Notice
being the maximum is a property of just the set of numbers, whereas
being a prefix-maximum is a property of the random sequence and the
current position. A \emph{wait-and-pick} algorithm is one that rejects
the first $m$ numbers, and then picks the first prefix-maximum
number.
  
\begin{theorem}
  \label{thm:secy} As $n \rightarrow \infty$, 
  the 37\%-algorithm picks the highest number with probability at least
  $\nicefrac1e$. Hence, it gets
  expected value at least $v_{\max}/e$. Moreover, $n/e$ is the optimal
  choice of $m$ among all wait-and-pick algorithms.
\end{theorem}

\begin{proof} 
  If we 
  pick the first prefix-maximum after rejecting the first $m$ numbers, the probability we pick the maximum is
  \begin{align*} 
    & \sum_{t = m+1}^n \Pr[ v_t \text{ is max} ] \cdot
    \Pr[ \text{max among first $t-1$ numbers falls in first $m$ positions}]
    \\   
    &\stackrel{(\star)}{=} \sum_{t = m+1}^n \frac1n \cdot \frac{m}{t-1} \quad = \quad
    \frac{m}{n} \big(H_{n-1} - H_{m-1}\big),
  \end{align*}
  where $H_k = 1+\frac12+\frac13+ \ldots+\frac1k$ is the $k^{th}$
  harmonic number. The equality $(\star)$ uses the uniform random order.
  Now using the approximation $H_k \approx \ln k + 0.57$ for large $k$, we get the
  probability of picking the maximum is about $\frac{m}{n} \ln
  \frac{n-1}{m-1}$ when $m,n$ are large.
  This quantity has a maximum value of $1/e$ if we choose $m = n/e$.
\end{proof}

Next we show 
we can replace any strategy (in a
comparison-based model) with a wait-and-pick strategy
without decreasing
the probability of picking the maximum.

\begin{theorem}
  \label{thm:secyTight}
  The strategy that maximizes the probability of
  picking the highest number can be assumed to 
  be a wait-and-pick strategy. 
\end{theorem}
\begin{proof} Think of yourself as a player trying to maximize the probability
of picking the maximum number. Clearly, you should reject the next number $v_i$
if  it is not prefix-maximum. Otherwise, you should pick $v_i$ only 
if  it is prefix-maximum 
and the probability of $v_i$ being the  maximum is more than the probability
of you picking the maximum in the remaining sequence. Let us calculate 
these probabilities.

We use Pmax to abbreviate ``prefix-maximum''. 
For position $i \in \{1, \ldots, n\}$,   define
 \begin{align*}
   f(i) &= \Pr[ v_i \text{ is max}\mid v_i \text{
          is Pmax} ] 
          \stackrel{(\star)}{=} \frac{\Pr[ v_i \text{ is max} ]
          }{\Pr[v_i \text{ is Pmax}]} \stackrel{(\star\star)}{=} \frac{1/n}{1/i} =
            \frac{i}n, 
 \end{align*}
 where equality $(\star)$ uses that the maximum is also a prefix-maximum,
 and $(\star\star)$ uses the uniform random ordering.  Note that $f(i)$
 increases with $i$.
 
 Now consider a problem where the numbers are again being revealed in a random
 order but we must reject the first $i$ 
 numbers. The goal is to still  maximize the probability of picking the highest of the $n$ numbers.  Let $g(i)$ denote the 
 probability that the optimal strategy for this problem picks the global maximum.

The function $g(i)$ must be a non-increasing function of $i$, else we could just ignore
the $(i + 1)^{st}$ number and set $g(i)$ to mimic the strategy for $g(i + 1)$. Moreover,
$f(i)$ is increasing. So from the discussion above, you should not pick a prefix-maximum 
number at any
position $i$ where $f(i) < g(i)$ since you can do better on the suffix. Moreover, when
$f(i) \geq g(i)$, you should pick $v_i$
if it is prefix-maximum, since it is
worse to wait. Therefore, the approach of waiting until $f$ becomes greater than $g$
and thereafter picking the first prefix-maximum is an optimal strategy. 
\end{proof}

Theorems~\ref{thm:secy} and~\ref{thm:secyTight} imply for $n \to \infty$
that no algorithm can pick the maximum with probability 
more than $1/e$. Since we placed no bounds on the number magnitudes, this can
 also be used to show that for any $\eps>0$, there exist an $n$
 and numbers $\{v_i\}_{i \in \{1,\ldots, n\}}$ 
 where every algorithm has expected value
at most $(1/e+\eps) \cdot \max_i v_i$.


  \saveforlong{
  A ``modern'' proof of this result uses linear-programming (LP). This
approach is robust and gives algorithms even when we add some kinds of
constraints to the problem. The approach involves writing down an LP
that captures some properties of any feasible solution, solving this LP
optimally, and converting this solution into an algorithm with success
probability equal to the objective value of this LP.

\begin{proof}
  Let us fix an optimal strategy. (By the proof above, we know what this
  strategy is, but let us ignore that for now.) We just assume that it
  does not pick any number that is not a local best, as such a
  number cannot be the global best either. 
  Let $p_i$ be the probability that this strategy picks the number at
  position $i$. Let $q_i$ be the probability that we pick the number at
  position $i$ \emph{conditioned on the number being a local best}. So
  $q_i = \frac{p_i}{1/i} = i\cdot p_i$.
  Now, the probability of picking the highest number is 
  \begin{align}
    & \ts\sum_i \Pr[ v_i \text{ is global best and we pick it} ]
    \notag \\
    &= \ts \sum_i \Pr[ v_i \text{ is global best} ] \cdot q_i
    \quad = \quad \sum_i \frac{1}{n} q_i \quad = \quad \sum_i \frac{i}{n} p_i. \label{eq:7}
  \end{align}
  What are the constraints? Clearly $p_i \in [0,1]$. But also
  \begin{align}
    p_i &= \Pr[ \text{ pick $v_i \mid v_i$ is local best} ] \cdot \Pr[
    v_i \text{ is local best} ] \notag \\
    &\leq \Pr[ \text{ did not pick any of } v_1, \ldots,v_{i-1} \mid v_i \text{
      is local best} ] \cdot (1/i). \label{eq:6}
  \end{align}
  Not picking the first $i-1$ numbers is independent of $i$ being
  local 
  best, so we can remove the conditioning to get 
  \begin{align*}
 	\textstyle{    p_i \leq \frac{1}{i} \cdot \big(1 - \sum_{j < i} p_j \big). } \label{eq:5} 
  \end{align*}
  Hence, the success probability of any strategy (and hence of the
  optimal strategy) is upper-bounded by the following LP in variables
  $p_i$:
  \begin{align*}
    \ts \max \quad \sum_i \frac{i}{n}\cdot p_i & \\
    \text{subject to  ~~} i\cdot p_i &\leq \ts 1 - \sum_{j < i} p_j \qquad \forall i \in \{1,
                 \ldots, n\}\\
    p_i &\in [0,1].
  \end{align*}
  It can be checked that the solution $p_i = 0$ for $i \leq \tau$
  and $p_i =\frac{\tau}{n}(\frac1{i-1} - \frac1i)$ for $\tau \leq i \leq
  n$ satisfies the constraints; here $\tau$ is defined by the smallest
  value such that $H_{n-1} - H_{\tau - 1} \leq 1$. (By duality, we can
  also show this is the optimal LP solution!)

  Here is how to convert this LP solution into an algorithm: if the
  current element $v_i$ is a local best, and we have not picked a number
  already, we pick $v_i$ with probability
  $\frac{i p_i}{1 - \sum_{j < i} p_j}$. (This expression lies in the
  interval $[0,1]$ by the LP constraint for $i$, so the algorithm is
  well-defined.)  Moreover, multiplying this expression with the
  probability that we have not picked an element earlier, and using the
  calculation in~(\ref{eq:6})) implies that we pick $v_i$ with
  probability $p_i$. Finally, a calculation similar to~(\ref{eq:7})
  shows that our algorithm's probability of picking the global best is
  $\sum_i ip_i/n$, which is the same as the LP value.
\end{proof}

As mentioned above, the LP approach is extendable and
robust, and it makes adding constraints easy. For example, note that the
optimal strategy does not pick anything in the first $n/e$ timesteps,
and then picks numbers with varying probabilities. If we imagine this
modeling people interviewing for a job, this gives candidates an
incentive to not come early in the order. To avoid this behavior, we may
want that for each position $i$, the probability of picking the item at
position $i$ is the same. By setting all $p_i = p$ and changing the
constraints slightly (e.g., we may now want to pick numbers even when
they are not the local best), this general approach extends to give an
algorithm picking $v_{\max}$ with probability
$1 - 1/\sqrt{2} \approx 29\%$.
}

\newcommand{\bp}{\mathbf{p}}

\section{Multiple-Secretary and Other Maximization Problems}
\label{sec:sec-ext}



We now extend our insights from the single-item case to
settings where we can pick multiple items. Each item has a value, and we
have constraints on what we can pick (e.g., we can  pick at most $k$
items, or pick any acyclic subset of 
edges of a graph). The goal is to maximize the total value. (We
study minimization problems in \S\ref{sec:comb}.)
Our algorithms can be broadly classified
as being \emph{order-oblivious} or \emph{order-adaptive}, depending on the
degree to which they rely on the random-order assumption.

\subsection{Order-Oblivious Algorithms}
\label{sec:order-obliv-algor}

The 50\%-strategy for the single-item secretary problem has an
interesting property: if each number is equally likely to lie in the
first or the second half, we pick $v_{\max}$ with
probability $\nicefrac14$ even if the arrival sequence within the
first and second halves is chosen by an adversary.  To formalize this
property, define an \emph{order-oblivious} algorithm as one with the
following two-phase structure: in the first phase (of some length
$m$) the algorithm gets a uniformly-random subset of $m$ items,
but is not allowed to pick any of these items. In the second phase, the
remaining items arrive in an adversarial order, and only now can the
algorithm pick items while respecting any constraints that 
exist. (E.g., in the secretary problem,  only one item may be picked.)
Clearly, any order-oblivious algorithm runs in the random-order model
with the same (or better) performance guarantee, and hence we can
focus our attention on designing such algorithms. Focusing on
order-oblivious algorithms has two
benefits. Firstly, such algorithms are easier to design and analyze,
which becomes crucial when the underlying constraints become more
difficult to reason about. Secondly, the guarantees of such algorithms can be 
interpreted as holding even for adversarial arrivals, as long as we have
offline access to some samples from the underlying distribution (discussed
in \S\ref{sec:others}).  To make things concrete, let us
start with the simplest generalization of the secretary problem.

\subsubsection{The Multiple-Secretary Problem: Picking $k$ Items}
\label{sec:k-items}

We now pick $k$ items to maximize the expected sum of their
values: the case $k=1$ is the secretary problem from the previous
section. We associate the items with the set $[n] = \{1, \ldots,
n\}$, with item $i \in [n]$ having value $v_i \in \R$; all values
are distinct.  Let $S^\star \sse [n]$ be the set of $k$ items of
largest value, and let the total value of the set $S^\star$ be
$V^\star := \sum_{i \in S^\star} v_i$.

It is easy to get an algorithm that gets expected value
$\Omega(V^\star)$, e.g., by splitting the input sequence of length $n$
into $k$ equal-sized portions and running the single-item algorithm
separately on each of these, or by setting threshold $\tau$ to be the
value of (say) the $\ceil{\nicefrac{k}3}^{th}$-highest value item in the
first $50\%$ of the items and picking the first $k$ items in the second
half whose values exceed $\tau$ (see
Exercise~\ref{exer:multItemConstant}). Since both these algorithms
ignore a constant fraction of the items, they lose at least a constant
factor of the optimal value in expectation. But 
we may hope to do better. Indeed, the $50\%$ algorithm obtains
a (noisy) estimate of the threshold between the maximum value item and
the rest, and then picks the first item above the threshold. The
simplest extension of this idea would be to estimate the threshold
between the top $k$ items, and the rest.  Since we are picking
$k \gg 1$ elements, we can hope to get accurate estimates of this
threshold by sampling a smaller fraction of the stream.

The following (order-oblivious) algorithm formalizes this intuition.
It gets an
expected value of $V^\star(1 - \delta)$, where $\delta \to 0$ as
$k \to \infty$. To achieve this performance, we get an accurate estimate of the
$k^{th}$ largest item in the entire sequence after  ignoring only
$\delta n$ items, and hence can start picking items much earlier.

\begin{algo}
  \begin{enumerate}
  \item Set $\eps = \delta = O\big( \frac{\log k}{k^{1/3}}\big)$.
  \item \emph{First phase:} ignore the first $\delta n$
    items. \\Threshold $\tau \gets$  
   value of the $(1-\eps)\delta k^{th}$-highest valued
  item in this ignored set.
  \item \emph{Second phase:} pick the first $k$ items seen that have value greater than $\tau$.
  \end{enumerate}
\end{algo}

\begin{theorem} \label{thm:kItemObiv} The  order-oblivious
  algorithm above for the multiple-secretary problem has expected value
  $V^\star (1 - O(\delta))$, where
  $\delta = O\big( \frac{\log k}{k^{1/3}}\big)$.
\end{theorem}

\begin{proof}
  The $\delta n$ items ignored in the first phase contain
   in expectation $\delta k$ items from $S^\star$, so we lose expected value
  $\delta V^\star$. Now a natural threshold would be the
  $\delta k^{th}$-highest value item among the ignored items. To account
  for the variance in how many elements from $S^\star$ fall among the
  ignored elements, we  set a slightly 
  higher threshold of the $(1-\eps)\delta k^{th}$-highest value.

  Let $v' := \min_{i \in S^\star} v_i$ be the lowest value item we
  actually want to pick.  There are two failure modes for this
  algorithm: (i)~the threshold is too low if $\tau < v'$, as then we may
  pick low-valued items, and (ii)~the threshold is too high if fewer
  than $k - O(\delta k)$ items from $S^\star$ fall among the last
  $(1-\delta)n$ items that are greater than $\tau$.  Let us see why
  both these bad events happen rarely.

  \begin{itemize} 
  \item {\em Not too low}: For event~(i) to happen, fewer than
    $(1-\eps)\delta k$ items from $S^\star$ fall in the first $\delta n$
    locations: i.e., their number is less than $(1-\eps)$ times its
    expectation $\delta k$. This has probability at most
    $\exp(-\eps^2 \delta k)$ by Chernoff-Hoeffding concentration bound 
	(see the aside below). Notice if $\tau \geq v'$ then we never run out of
	budget $k$.
  \item {\em Not too high}: For event~(ii), let $v''$ be 
    the $(1-2\eps)k^{th}$-highest value in $S^\star$. We expect
    $(1- 2\eps)\delta k$  items above $v''$ to appear among the
    ignored items, so the probability that more than $(1-\eps) \delta k$
    appear is $\exp(-\eps^2 \delta k)$ by Chernoff-Hoeffding 
    concentration bound.  This
    means that $\tau \leq v''$ with high probability, and moreover most
    of the high-valued items appear in the second phase (where we will
    pick them whenever event~(i) does not happen, as we don't run 
    out of budget).
  \end{itemize}
  Finally, since we are allowed to lose $O(\delta V^\star)$ value,
  it suffices that the error probability $\exp(-\eps^2 \delta k)$
  be at most $O(\delta) = 1/\poly(k)$. This requires us to set $\eps^2
  \delta k = \Omega(\log k)$, and a good choice of parameters is $\eps = \delta = O\big( \frac{\log k}{k^{1/3}}\big)$.
\end{proof}

An aside: the familiar Chernoff-Hoeffding concentration bounds 
(Exercise~1.3(a) in Chapter~8) are for
sums of bounded \emph{independent} random variables, but the RO model has
 correlations (e.g., if one element from $S^\star$
falls in the first $\delta n$ locations, another is slightly
less likely to do so). The easiest fix to this issue is to ignore not
the first $\delta n$ items, but instead a random number of items with
the number drawn from a $\text{Binomial}(n,\delta)$ distribution with
expectation $\delta n$. In this case each item has probability $\delta$
of being ignored, independent of others.
A second way to achieve independence is to imagine each arrival 
happening at a  uniformly and independently chosen time in $[0,1]$. Algorithmically, we can
 sample $n$ i.i.d.\ times from $\text{Uniform}[0,1]$, sort them in 
increasing order, and assign the $i^{th}$ time to the $i^{th}$ arrival.
Now, rather than ignoring
the first $\delta n$ arrivals, we can ignore arrivals happening before 
time $\delta \in [0,1]$.
Finally, a third alternative is to not strive for independence, but 
instead directly 
use concentration bounds for sums of exchangeable random variables.
Each of these three alternatives offers different
benefits, and one alternative might be much easier
to analyse than the others, depending on
 the  problem at hand.

The loss of $\approx V^\star/k^{1/3}$ in Theorem~\ref{thm:kItemObiv}
is not optimal. We will see an
order-adaptive algorithm in the next section which achieves an expected value
of $V^\star \big(1 - O\big(\sqrt{\nicefrac{\log k}{k}}\big) \big)$. That algorithm
will not use a single threshold, instead  will adaptively refine its
threshold as it sees more of the
sequence.  
But first, let us discuss a few more order-oblivious
algorithms for other combinatorial constraints.


\subsubsection{Maximum-Weight Forest}
\label{sec:forest}

Suppose the items arriving in a random order are the $n$ edges of a
(multi-)graph $G=(V,E)$, with edge $e$ having a value/weight
$v_e$. The algorithm knows the graph at the beginning, but not the
weights. When the edge $e$ arrives, its weight $v_e$ is revealed, and we
decide whether to pick the edge or not. Our goal is to pick a subset of
edges with large total weight that form a forest (i.e., do not contain a cycle). The
target $V^\star$ is the total weight of a maximum-weight 
forest of the graph: offline, we can solve this problem using, e.g., 
Kruskal's greedy algorithm. This \emph{graphical secretary} problem generalizes
the secretary problem: Imagine a graph with two vertices and
$n$ parallel edges between them. Since any two edges
form a cycle, we can pick at most one edge, which models the
single-item problem.

As a first step towards an algorithm, suppose all the
edge values are either $0$ or $v$ (but we don't know in advance which
edges have what value). A greedy algorithm is to pick the next weight-$v$ edge whenever possible, i.e.,
when it does not create cycles with previously picked edges. This
returns
a max-weight
forest, because the optimal solution is a maximal forest among the subset of
weight-$v$ edges, and every maximal forest in a graph has the same
number of edges. 
This suggests
the following algorithm for general values: if we know some value $v$ for which there is a
subset of acyclic edges, each of value $v$, with total weight
$\geq \frac{1}{\alpha} \cdot V^\star$, then we can get an
$\alpha$-competitive solution by greedily picking value-$v$ edges
whenever possible.

How do we find such a value $v$ that gives a good approximation?
The Random-Threshold algorithm below uses
two techniques: \emph{bucketing} the values and \emph{(randomly) mixing}
a collection of algorithms.  We assume that all values are powers of
$2$; indeed, rounding values down to the closest power of $2$ loses at
most a factor of $2$ in the final guarantee. 

\begin{algo}
  \begin{enumerate}
  \item Ignore the first $n/2$ items and let $\hat{v}$ be their highest
    value.  
  \item Select a uniformly random $r \in \{0,\ldots, \log n\}$, and set
    threshold $\tau := \hat{v}/2^r$.
  \item For the second $n/2$ items, greedily pick any item of value at
    least $\tau$ that does not create a cycle.
  \end{enumerate}
\end{algo}

\begin{theorem}
  \label{thm:maxWtdForestLog}
  The  order-oblivious  Random-Threshold algorithm for the graphical secretary problem 
  gets  an expected value  $\Omega \big(\frac{V^\star}{\log n}\big)$.
\end{theorem}

\onlyforshort{
Here is the main proof idea: either most of the value is in a
single item (say $v_{\max}$), in which case when $r=0$ 
(with probability $1/\log n$) this mimics
 the $50\%$-algorithm. 
 Else, we can assume that $v_{\max}$ falls in the first half, giving
us a good estimate without much loss. Now, very little of $V^\star$ can
come from items of value less than $v_{\max}/n$ (since there are only
$n$ items). So we can focus on $\log n$ buckets of items whose values
lie in $[v_{\max}/2^{i+1}, v_{\max}/2^i)$. These buckets, on average,
contain value $V^*/\log n$ each, and hence picking a random one does well.
}

\saveforlong{
\begin{proof}
  What does the optimum solution $S^\star$ (with value $V^\star$)
  look like? Let $v_{\max}$ be the maximum value of any edge/item
  in $S^{\star}$. If $v_{\max} \geq V^\star/4$, we could run
  the single-item order-oblivious algorithm to get an expected value
  of $v_{\max}/4 = \Omega(V^\star)$.

  Else, define ``bucket'' $B_i$ to be edges of value exactly
  $v_{\max}/2^i$.  We claim that the items in $S^\star$ that lie within
  buckets $B_0, B_1, \ldots, B_{\log_2 n}$ have total value at least
  $\nicefrac{V^\star}{2}$. Indeed, if not, 
  at least $\nicefrac{V^\star}{2}$ of $S^\star$'s value must come from
  items of value less than $v_{\max}/2^{\log_2 n} = v_{\max}/n$. But
  there are at most $n-1$ edges in the forest $S^\star$. Even if each
  of those edges contributes  $v_{\max}/n < V^\star/(4n)$, we cannot get
  $V^\star/2$ value from them, a contradiction. Hence, we get that 
  \begin{equation}  \label{eq:maxWtForest}
  \textstyle{ \sum_{i=0}^{\log n} ({v_{\max}}/{2^i}) \cdot |S^\star \cap B_i| \quad \geq \quad
  {V^\star}/{2}.}
  \end{equation}
  Moreover, for any of these buckets $i$,
   we could run the algorithm that picks items of
  value exactly $v_{\max}/2^i$, and get a value of 
  ${v_{\max}}/{2^i} \cdot |S^\star \cap B_i|$.

  The difficulty is that we know neither $v_{\max}$, nor whether
  $v_{\max} \geq V^\star/4$, nor which of these buckets to choose. Hence
  the idea is to ``mix'' all the above algorithms by randomly choosing
  between them.
 
  It is easy to see that if $v_{\max} \geq V^\star/4$, the single-item
  algorithm does well. Else, condition on $v_{\max}$ is in the first
  half (which happens with probability $\nicefrac12$), so 
  $\hat{v} = v_{\max}$. Moreover, for any choice of
  $r=i \in \{0,\ldots, \log n\}$ (which happens with probability
  $\Theta(\nicefrac{1}{\log n})$), the second half contains in expectation half the
  elements of $S^\star \cap B_i$; now
  using~\eqref{eq:maxWtForest}, our algorithm gets value
  $\Omega(V^\star/\log n)$ in expectation. 
\end{proof}
}

\paragraph{An Improved Algorithm for Max-Weight Forests.}
The Random-Threshold algorithm above used relatively few properties of the max-weight forest.
Indeed, it extends to 
downward-closed set systems with the property that if all values are $0$ or $v$ then picking the next
 value-$v$ element whenever possible gives a near-optimal solution. 
However, we can do better using properties 
of the underlying graph. Here is a constant-competitive algorithm for graphical secretary 
where the main idea is to decompose the
problem into several disjoint single-item secretary problems. 

\begin{algo}
  \begin{enumerate}
  \item Choose a uniformly random permutation $\widehat{\pi}$
    of the vertices of the graph.
  \item For each edge $\{u,v\}$, direct it from $u$ to $v$ if
    $\widehat{\pi}(u)<\widehat{\pi}(v)$. 
  \item Independently for each vertex $u$, consider the edges directed
    \emph{towards} $u$ and run the order-oblivious $50\%$-algorithm  on these edges.
  \end{enumerate}
\end{algo}

\begin{theorem}
  \label{thm:maxWtdForestO1}
  The algorithm above for the graphical secretary problem is order-oblivious and gets
  an expected value at least $\nicefrac{V^\star}{8}$. 
\end{theorem}
\begin{proof}
  The algorithm picks a forest, i.e., there
  are no cycles (in the undirected sense) among the picked edges.
  Indeed, the highest numbered vertex (w.r.t. $\widehat{\pi}$) on any
  such cycle would have two or more  incoming edge picked, which is not
  possible.
    
  However, since we restrict to picking only one incoming edge per vertex, the
  optimal max-weight forest $S^\star$ may no longer be
  feasible. Despite this, we claim there is a 
  forest with the one-incoming-edge-per-vertex restriction, and expected value
  $V^\star/2$. (The randomness here is over the choice of the permutation
  $\widehat{\pi}$, but not of the random order.) Since the $50\%$-algorithm gets a quarter of this
  value (in expectation over the random  ordering), we get the
  desired bound of $V^\star/8$.

  \begin{figure}[h]
    \begin{center}
      \includegraphics[width=0.9\textwidth]{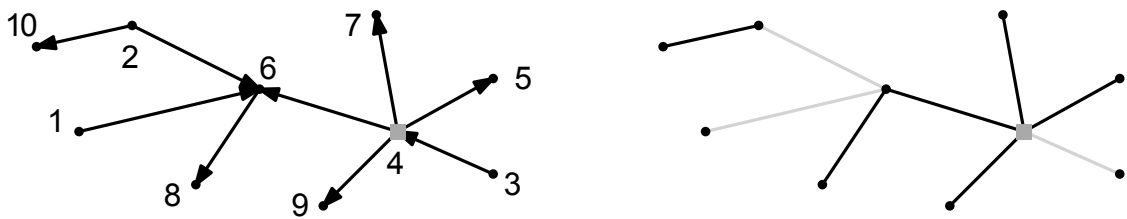}
      \caption{The optimal tree: the numbers on the left are those given
      by $\widehat{\pi}$. The grey box numbered $4$ is the root. The edges in the right are those retained in
      the claimed solution.}\label{fig:graphical}
    \end{center}
  \end{figure}

  To prove the claim, root each component of $S^\star$ at an
  arbitrary node, and associate each non-root vertex $u$ with the unique
  edge $e(u)$ of the undirected graph on the path towards the root. The proposed solution
  chooses for each vertex $u$, the edge $e(u) = \{u,v\}$ if
  $\widehat{\pi}(v) < \widehat{\pi}(u)$, i.e., if it is
  directed \emph{into} $u$ (Figure~\ref{fig:graphical}). Since this event happens
  with probability $1/2$, the proof follows by linearity of
  expectation. 
\end{proof}

  This algorithm is order-oblivious because the  $50\%$-algorithm has the property.
  If we don't care about order-obliviousness, we can instead use the 
  $37\%$-algorithm and get expected value at least $\nicefrac{V^\star}{2e}$.

\subsubsection{The Matroid Secretary Problem}

One of the most tantalizing generalizations of the secretary problem is
to \emph{matroids}. (A matroid defines a notion of \emph{independence}
for subsets of elements, generalizing linear-independence of a
collection of vectors in a vector space. E.g., if we define a subset of edges to be
independent if they are acyclic, these form a ``graphic'' matroid.)
Suppose the $n$ items form the ground set elements of a known matroid,
and we can pick only subsets of items that are independent 
in this matroid. The weight/value $V^\star$ of the max-weight independent set
can be computed offline by the obvious generalization of Kruskal's
greedy algorithm.
The open question is to get an expected value of $\Omega(V^\star)$ online in
the RO model. The approach from Theorem~\ref{thm:maxWtdForestLog} gives
expected value $\Omega(V^\star/\log k)$, where $k$ is the largest size of an
independent set (its \emph{rank}).  The current best algorithms
(which are also order-oblivious) achieve an expected value of
$\Omega(V^\star/\log \log k)$.  Moreover, we can obtain $\Omega(V^\star)$
for many special classes of matroids by exploiting their special
properties, like we did for the graphic matroid above; see the notes for
references.


\subsection{Order-Adaptive Algorithms} \label{sec:orderAdap}

The order-oblivious algorithms above have several benefits, but their
competitive ratios are often worse than  \emph{order-adaptive}
algorithms where we exploit the randomness of the entire arrival order. 
Let us revisit the problem of picking $k$ items.

\subsubsection{The Multiple-Secretary Problem Revisited}
\label{sec:multi-secy-II}

In the order-oblivious case of Section~\ref{sec:k-items}, we ignored the
first $\approx k^{-1/3}$ fraction of the items in the first phase, and
then chose a fixed threshold to use in the second phase. The length of
this initial phase was chosen to balance two competing concerns: we
wanted the first phase to be short, so that we ignore few items, but we
wanted it to be long enough to get a good estimate of the $k^{th}$
largest item in the entire input. The idea for the improved algorithm is to
run in multiple phases and 
use time-varying thresholds. Roughly, the algorithm uses the first 
$n_0 = \delta n$ arrivals to
learn a threshold for the next $n_0$ arrivals, then it computes a new
threshold at time $n_1 = 2 n_0$ for the next $n_1$ arrivals, and so on.

As in the order-oblivious algorithm, we aim for the
$(1-\eps)k^{th}$-largest element of $S^\star$---that $\eps$ gives us a
margin of safety so that we don't pick a threshold lower than the
$k^{th}$-largest (a.k.a.\ smallest) element of $S^\star$. But we vary
the value of $\eps$. In the beginning, we have low confidence,
so we pick items  cautiously (by setting a high $\eps_0$ and
creating a larger safety margin).  As we see more elements, we are more
confident in our estimates, and can decrease the $\eps_j$ values.

\begin{algo}
  \begin{enumerate}
  \item   Set $\delta := \sqrt{\frac{\log k}{k}}$.  Denote  $n_j := 2^j\delta n$
  and ignore first $n_0 = \delta n$ items.
  \item  For $j \in [0,\log \nicefrac1\delta]$, phase $j$ runs on 
  arrivals in \emph{window} $W_j := (n_{j}, n_{j+1}]$
 	 \begin{itemize}
  \item  
  Let $k_j := (k/n) n_j$ and   let $\eps_j := \sqrt{\delta/2^j}$.
  \item    Set   threshold $\tau_j$ to be the  $(1-\eps_j)k_j^{th}$-largest
  value among the first $n_j$ items. 
  \item  Choose
  any item in  window $W_j$ with value above $\tau_j$ (until budget $k$ is
  exhausted).
  	\end{itemize}
  \end{enumerate}
\end{algo}

\begin{theorem} \label{thm:kItemsAdap} The  order-adaptive 
algorithm above for the
  multiple-secretary problem has expected value
  $V^\star \big(1 - O\big(\sqrt{\frac{\log k}{k}}\big) \big)$. 
\end{theorem}
\begin{proof}
\ignore{
  The algorithm runs in phases where  for 
  $j \in [0,\log \nicefrac1\delta]$ and $n_j := 2^j\delta n$, 
  Phase~$j$ runs on a \emph{window} $W_j := (n_{j}, n_{j+1}]$.
  Let $k_j := 2^j \delta k$ so that we expect to
  see $k_j$ items from the optimal set $S^\star$ in the first $n_j$
  items. Let $\eps_j := \sqrt{\delta/2^j}$. After seeing $n_j$ items, the
  threshold $\tau_j$ is the value of the $(1-\eps_j)k_j^{th}$-largest
  item among these $n_j$ items. The algorithm ignores the items
  before $n_0$, and then in each subsequent window $W_j$ it chooses
  any items with value at least $\tau_j$.
}

\ignore{
  Set $\delta := \sqrt{\frac{\log k}{k}}$ and $\eps := \delta/2$. For
  integers $j \in [0,\log \nicefrac1\delta]$, define
  $n_j := 2^j\delta n$ and $k_j := 2^j \delta k$ so that we expect to
  see $k_j$ items from the optimal set $S^\star$ in the first $n_j$
  items. Define the $j^{th}$ \emph{window} $W_j := (n_{j}, n_{j+1}]$,
  and let $\eps_j := \sqrt{\delta/2^j}$. After seeing $n_j$ items, the
  threshold $\tau_j$ is the value of the $(1-\eps_j)k_j^{th}$-largest
  item among these $n_j$ items. The algorithm ignores the items
  before $n_0$, and then in each subsequent window $W_j$ it chooses
  any items with value at least $\tau_j$.
}

  As in Theorem~\ref{thm:kItemObiv}, we first show that none of the
  thresholds $\tau_j$ are ``too low'' (so  we never run out of budget $k$). 
  Indeed, for $\tau_j$ to lie below $v' := \min_{i \in S^\star} v_i$,
   less than $(1-\eps_j)k_j$ items from $S^\star$ should fall in
  the first $n_j$ items. Since we expect $k_j$ of them, the probability of
  this is at most $\exp(-\eps_j^2 k_j) = \exp(-\delta^2 k) = 1/\poly(k)$.

  Next, we claim that $\tau_j$ is not ``too high'': it is
  with high probability at most the
  value of the $(1 - 2 \eps_j) k^{th}$ highest item in $S^\star$
  (thus all thresholds are at most $(1 - 2 \eps_0) k^{th}$ highest value).
 Indeed, we expect $(1- 2 \eps_j) k_j$ of these
  highest items
  to appear in the first $n_j$ arrivals, and the probability that more
  than $(1-\eps_j) k_j$ appear is $\exp(-\eps_j^2 k_j) = 1/\poly(k)$.
    
    Taking a union bound over all $j \in [0,\log \nicefrac1\delta]$, with high 
  probability  all thresholds are neither too high nor too low. Condition on
  this good event. 
  Now  any of the top $(1-2\eps_0)k$ items will be picked if it
  arrives after first the $n_0$ arrivals (since no threshold is too high
  and we never run out of budget $k$), 
  i.e., with probability $(1- \delta)$. Similarly, any item
  which is in top $(1-2 \eps_{j+1})k$,  but not in top $(1-2 \eps_{j})k$, will be 
  picked if it arrives after $n_{j+1}$, i.e., with probability $(1- 2^{j+1}\delta)$. Thus 
   if $v_{\max} = v_1 > \ldots > v_k$ are the top $k$ items, we get
  an expected value of 
    \begin{gather*}
   \sum_{i = 1}^{(1-2\eps_0)k} v_i (1-\delta) + \sum_{j = 0}^{\log
      \nicefrac1\delta -1} ~\sum_{i = (1-2\eps_j)k}^{(1-2\eps_{j+1})k} v_i
    (1-2^{j+1}\delta).
  \end{gather*}
  This is at least $V^\star(1-\delta) - \nicefrac{V^\star}k  \big( \sum_{j = 0}^{\log
      \nicefrac1\delta } 2\eps_{j+1}k \cdot 2^{j+1}\delta \big)$ because
       the negative terms are maximized when the
  top $k$ items are all equal to $\nicefrac{V^\star}k$. Simplifying, we get
  $V^\star(1- O(\delta))$, as claimed.
\end{proof}

The logarithmic term in $\delta$
can be removed (see the notes), but the loss of $\sqrt{k}$ is essential.
Here is a sketch of the lower bound. By Yao's minimax lemma, it suffices
to give a distribution over instances that causes a large loss for any
deterministic algorithm.
Suppose each item has value $0$ with
probability $1- \frac{k}{n}$, else it has value $1$ or $2$ with equal
probability. The number of non-zero items is therefore
$k \pm O(\sqrt{k})$ with high probability, with about half $1$'s and half
$2$'s, i.e., $V^\star = 3k/2 \pm O(\sqrt{k})$. Ideally, we want to pick all the $2$'s and then
fill the remaining $k/2 \pm O(\sqrt{k})$ slots using the $1$'s. However, consider
 the state of the algorithm after $n/2$ arrivals. Since the
algorithm doesn't know how many $2$'s will arrive in the second half, it
doesn't know how many $1$'s to pick in the first half. Hence, it will
either lose about $\Theta(\sqrt{k})$ $2$'s in the second half, or it will
pick $\Theta(\sqrt{k})$ too few $1$'s from the first half. Either way, the
algorithm will lose $\Omega(V^\star/\sqrt{k} )$ value.


\subsubsection{Solving Packing Integer Programs} \label{sec:LPsRandOrder}

The problem of picking a max-value subset of $k$ items 
can
be vastly generalized. Indeed, if each item $i$ has size
$a_i \in [0,1]$ and we can pick items having total size $k$, we get the knapsack
problem. More generally, suppose we have $k$
units each of $d$
different resources, and item $i$ is specified by the amount
$a_{ij} \in [0,1]$ of each resource $j \in \{1,\ldots,d\}$ it uses if
picked; we can pick any subset of items/vectors which can be supported by our
resources.
Formally, a set of $n$ different $d$-dimensional vectors
$\ba_1, \ba_2, \ldots, \ba_n \in [0,1]^d$ arrive in a random order, each
having an associated value $v_i$.  We can pick any
subset of items subject to the associated vectors summing to at most $k$ in each
coordinate.  (All vectors and values are initially unknown,
and on arrival of a vector we must irrevocably pick or
discard it.) We want to maximize the expected
value of the picked vectors. This gives 
a packing
integer program:
\begin{equation*}
\max \sum_i v_i x_i \quad \text{s.t.} \quad \sum_i x_i \ba_i
  \leq k \cdot \mathbf{1} \text{~~and~~} x_i \in \{0,1\}, \label{eq:packlp}
\end{equation*}
where the vectors arrive in a random order.  Let
$V^\star := \max_{\mathbf{x} \in \{0,1\}^d}\{ \mathbf{v}\cdot \mathbf{x} \mid A\mathbf{x} \leq k
\mathbf{1}\}$ be the optimal value, where $A \in [0,1]^{d\times n}$ has
columns $\ba_i$. The multiple-secretary problem is modeled by
$A$ having a single row of all ones.
By extending the approach from Theorem~\ref{thm:kItemsAdap}
considerably, one can achieve a competitive ratio of
$(1 - O(\sqrt{(\log d)/k}))$. In fact, several algorithms using
varying approaches are known, each giving the above competitive ratio.

We now sketch a weaker result.
To begin, we 
allow the variables to be fractional ($x_i \in [0,1]$) instead of
integers ($x_i \in \{0,1\}$). 
Since we assume the capacity $k$ is much larger than
$\log d$, we can use randomized rounding to go back 
from fractional to integer solutions with a little loss in value.
One of the key ideas is that 
learning a threshold
can be viewed as learning optimal dual values for this linear
program.

\begin{theorem}
  \label{thm:LPs}
  There exists an algorithm to solve packing LPs in the RO model to
  achieve expected value $V^\star \big(1 - O( \sqrt{\frac{d\log n}{k}}) \big)$.
\end{theorem}
\begin{proof}(Sketch) The proof is similar to that of
  Theorem~\ref{thm:kItemsAdap}. The algorithm uses windows of
  exponentially increasing sizes and (re-)estimates the optimal duals in
  each window. Let $\delta := \sqrt{\frac{d\log n}{k}}$; we will
  motivate this choice soon. As before, let $n_j = 2^j \delta n$,
  $k_j = (k/n) n_j$, $\eps_j = \sqrt{\delta/2^j}$, and the window
  $W_j = (n_j, n_{j+1}]$.  Now, our thresholds are the $d$-dimensional
  optimal dual variables $\pmb{\tau}_j$ for the linear program:
  \begin{equation} 
    \textstyle {\max \sum_{i=1}^{n_j} v_i x_i \quad \text{s.t.} \quad
    \sum_{i=1}^{n_j} x_i \ba_i \leq (1-\eps_j)k_j \cdot \mathbf{1}
    \text{~~and~~} x_i \in [0,1].} \label{eq:LPAdap}
  \end{equation}
  Having computed $\pmb{\tau}_j$ at time $n_j$, the algorithm picks an
  item $i\in W_j$ if $v_i \geq \pmb{\tau}_j \cdot \ba_i$.  In the
  $1$-dimensional multiple-secretary case, the dual is just the value of
  the $(1-\eps_j)k_j^{th}$ largest-value item among the first $n_j$,
  matching the choice in Theorem~\ref{thm:kItemsAdap}.  In general, the
  dual $\pmb{\tau}_j$ can be thought of as the
  price-per-unit-consumption for every resource; we want to select $i$
  only if its value $v_i$ is more than the total price
  $\pmb{\tau}_j \cdot \ba_i$.

  Let us sketch why the dual vector $\pmb{\tau}_j$ is not ``too low'':
  i.e., the dual $\pmb{\tau}_j$ computed is (with high probability) such
  that the set
  $\{ \ba_i \mid \pmb{\tau}_j \cdot \ba_i \leq v_i, i \in [n] \}$ of all
  columns that satisfy the threshold $\pmb{\tau}_j$ is still feasible.
  Indeed, suppose a price vector $\pmb{\tau}$ is \emph{bad}, and using
  it as threshold on the entire set causes the usage of some resource to
  exceed $k$. If $\pmb{\tau}$ it is an optimal dual at time $n_j$, the
  usage of that same resource by the first $n_j$ items is at most
  $(1-\eps_j)k_j$ by the LP~\eqref{eq:LPAdap}.  A Chernoff-Hoeffding
  bound shows that this happens with probability at most
  $\exp(-\eps_j^2 k_j) = o(1/n^d)$, by our choice of $\delta$. Now the crucial idea is to prune
  the (infinite) set of dual vectors to at most $n^d$ by only
  considering a subset of vectors using which the algorithm makes 
  different decisions. Roughly, there are $n$ choices of prices in each
  of the $d$ dimensions, giving us $n^d$ different possible bad dual
  vectors; a union-bound now gives the proof.
\end{proof}




As mentioned above, a stronger version of this result has an additive loss
$O(\sqrt{\frac{\log d}{k}}) V^\star$. Such a result is interesting only
when $k \gg \log d$, so this is called the ``large budget''
assumption. How well can we solve packing problems without such an
assumption? 
Specifically, given a downwards-closed family $\calF \subseteq 2^{[n]}$,
suppose we want to pick a subset of items having high total value and
lying  in $\calF$. 
For the
information-theoretic question where computation is not a concern, the
best known upper bound is $\Omega(V^\star/\log^2 n)$, and  
there
are families where $\Omega(V^\star/\frac{\log n}{\log\log n})$ is not
possible (see Exercise~\ref{exer:packLowerBound}). Can we close this
gap?  Also, which downward-closed families $\calF$ admit efficient
algorithms with good guarantees?




\subsubsection{Max-Weight Matchings}
\label{sec:combAuctionsRandom}

Consider a bipartite graph on $n$ agents and $m$ items. 
Each agent $i$ has a value $v_{ij} \in \mathbb{R}_{\geq 0}$ for
item $j$. The maximum-weight matching problem is to
 find an assignment
$M:[n] \rightarrow [m] $ to maximize $\sum_{i\in [n]} v_{iM(i)}$
such that no item $j$ is assigned to more than
one agent, i.e., $|M^{-1}(j)| \leq 1$ for all $j \in [m]$. 
In the online setting, which
has applications to allocating advertisements, the $m$ items are given
up-front and the $n$ agents arrive one by one. Upon arriving, agent $i$
reveals their valuations $v_{ij}$ for $j \in [m]$, whereupon we may irrevocably
allocate one of the remaining items to $i$. 
Let $V^\star$ denote the value of the optimal matching. The case of
$m=1$ with a single item is exactly the single-item secretary problem.

\ignore{
Consider a bipartite graph on $n$ agents and $m$ items. Each agent
$i \in [n]$ has a \emph{valuation} function
$f_i: 2^{[m]} \rightarrow \mathbb{R}_{\geq 0}$ for subsets of items. The
\emph{welfare maximization} problem is to partition the items into $n$
sets where agent $i$ is allocated set $S_i$ of items; the goal is to
maximize the total \emph{welfare} $\sum_i f_i(S_i)$. (See the notes for
connections to combinatorial auctions.)  In the online setting, which
has applications to allocating advertisements, the $m$ items are given
up-front and the $n$ agents arrive one by one. Upon arriving, agent $i$
reveals their valuation function $f_i$, whereupon we need to irrevocably
allocate a subset $S_i$ of the remaining items to them. Let $V^\star$ denote the optimal welfare. The case of
$m=1$ with a single item is exactly the single-item secretary problem.

For simplicity, we study the practically interesting and computationally
tractable case of \emph{unit-demand} functions, where each agent's value
for a set $S$ is given by their value for the best single item in
$S$. I.e., there are ``edge'' values $v_{ij} \in \mathbb{R}_{\geq 0}$
for $i\in [n]$ and $j\in [m]$ such that $f_i(S) = \max_{j\in S}
v_{ij}$. In this case, we only assign singleton sets to players,
and the ideal solution is a max-value matching in the bipartite graph.
}

 
The main algorithmic technique in this section almost seems na\"{\i}ve
at first glance: after ignoring the few first arrivals, we make each
subsequent decision based on an optimal solution of the arrivals until
that point, \emph{ignoring} all past decisions. For the matching
problem, this idea translates to the following:

\begin{algo}
Ignore the first $n/e$ agents. When agent $i \in (n/e,n]$
arrives:
\begin{enumerate}
\item Compute a max-value matching $M^{(i)}$ for the first $i$
  arrivals (ignoring past decisions).
\item If $M^{(i)}$ matches the current agent $i$ to item $j$, and if $j$
  is still available, then allocate $j$ to agent $i$; else, give
  nothing to agent $i$.
\end{enumerate}
\end{algo}
 (We assume that the matching $M^{(i)}$ depends only on the
identities of the first $i$ requests, and is independent of their
arrival order.)
We show the power of
this idea by proving that it gives optimal competitiveness for
matchings.
\begin{theorem}
\label{thm:MWM-basic}
The algorithm gives a matching with expected value at least $V^\star/e$.
\end{theorem}
\begin{proof}
  There are two insights into the proof. The first is that the matching
  $M^{(i)}$ is on a random subset of $i$ of the $n$ requests, and so has
  an expected value at least $(i/n) V^\star$. The $i^{th}$
  agent is a random one of these and so gets expected value $V^\star/n$.

  The second idea is to show, like in Theorem~\ref{thm:secy}, that if
  agent $i$ is matched to item $j$ in $M^{(i)}$, then $j$ is
  free with probability $\frac{n}{e i}$. Indeed, condition on the set,
  but not the order, of first $i$ agents (which fixes $M^{(i)}$) and
  the identity of the $i^{th}$ agent (which fixes $j$). Now for any
  $k \in (n/e, i)$, the item $j$ was allocated to the $k^{th}$ agent
  with probability at most $\frac{1}{k}$ (because even if $j$ is matched
  in $M^{(k)}$, the probability of the corresponding agent being the
  $k^{th}$ agent is at most $1/k$). The arrival order of the first $k-1$
  agents is irrelevant for this event, so we can do this 
  argument for all   $s<k$: the probability 
   $j$ was allocated to the $s^{th}$ agent, conditioned on $j$ not
   being allocated to the $k^{th}$ agent,
	is at most $\frac{1}{s}$. 
  So the probability that $j$ is available for agent $i$ is at least
  $\prod_{n/e<k<i} \big( 1 - \frac{1}{k}\big) \approx \frac{n}{e i}$.
  Combining these two ideas and using linearity of expectation, the
  expected total matching value is at least
$       \sum_{i=1+n/e}^n \big( \nicefrac{n}{e i} \cdot \nicefrac{V^\star}{n} \big) \approx \nicefrac{V^\star}{e}. 
$
  \end{proof}

  This approach can be extended to combinatorial auctions where
  each agent $i$ has a submodular (or an XOS) valuation $v_i$
  and can be assigned a subset $S_i \subseteq [m]$ of items;
  the goal is to maximize total welfare $\sum_i v_i(S_i)$. 
  \ignore{    submodular and XOS valuation
  functions {(see notes)}: such functions have a property that
  after fixing any allocation $S_1\cup\ldots \cup S_n$, there exist
  natural \emph{linear} ``support'' prices $\bf{v}_i$ that satisfy
  $f_i(S_i) = \sum_{j\in S_i} v_{ij}$, which essentially reduces the
  analysis to singleton sets.  }
  Also, this approach of following the
  current solution (ignoring past decisions) extends to solving
  packing LPs: the algorithm solves a slightly scaled-down version of the current LP 
  at each step $i$ and sets the variable $x_i$
  according to the obtained solution.





\section{Minimization Problems}
\label{sec:comb}

We now study  \emph{minimization} problems in the RO
model. In these problems the goal is to {minimize} some notion of cost
(e.g., the length of augmenting paths, or the number of bins) subject to
fulfilling some requirements. All the algorithms in this section are
order-adaptive. We use $OPT$ to denote both the optimal solution on the
instance $S$ and its cost.


\subsection{Minimizing Augmentations in Online Matching}
\label{sec:matching}

We start with one reason why  the RO model  might help for online discrete
minimization problems. Consider a problem to be ``well-behaved'' if
there is always a solution of cost $\approx OPT$ to serve the remaining
requests. This is clearly true at the beginning of the input sequence,
and we want it to remain true over time---i.e., poor choices in the past
should not cause the optimal solution on the remaining requests to
become much more expensive. Moreover, suppose that the problem cost is
``additive'' over the requests. Then satisfying the next request, which
by the RO property is a random one of $i$ remaining requests, should
cost $\approx OPT/i$ in expectation. Summing over all $n$ requests gives
an expected cost of
$\approx OPT(\frac1n + \frac{1}{n-1} + \ldots + \frac12 + 1) = O(\log
n)\, OPT$. (This general idea is reminiscent of that for max-weight
matchings from \S\ref{sec:orderAdap}, albeit in a
minimization setting.)


To illustrate this idea, we consider an online bipartite matching problem.
Let $G = (U, V, E)$ be a bipartite graph with $|U|=|V|=n$. Initially
the algorithm does not know the edge set $E$ of the graph, and hence 
the initial matching $M_0 = \emptyset$. At each time step
$t \in [1, n]$, all edges incident to the $t^{th}$ vertex
$u_t \in U$ are revealed. If the previous matching $M_{t-1}$ is no
longer a maximum matching among the current set of edges, the algorithm
must perform an augmentation to obtain a maximum matching $M_t$. We do
not want the matchings to change too drastically, so we define the cost
incurred by the algorithm at step $t$ to be the length of the augmenting
path $M_{t-1} \triangle M_t$. The goal is to minimize the total cost of
the algorithm. (For simplicity, assume $G$ has a perfect
matching and $OPT = n$, so we need to augment at each step.)
A natural algorithm  is \emph{shortest augmenting path} (see Figure~\ref{fig:SAP}):

\begin{algo}
  When a request $u_t \in U$ arrives:
  \begin{enumerate}
  \item Augment along a shortest alternating path $P_t$ from $u_t$ to some
    unmatched vertex in $V$.
  \end{enumerate}
\end{algo}

  \begin{figure}[t]
    \begin{center}
      \includegraphics[width=0.6\textwidth]{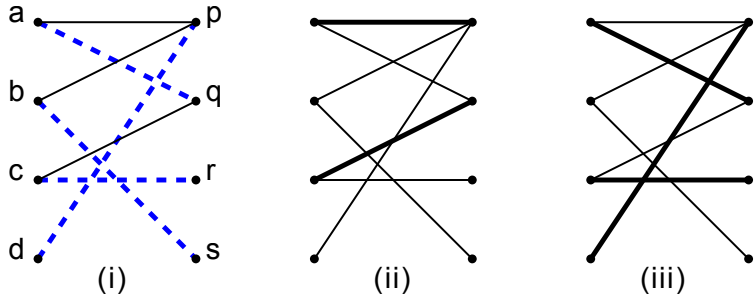}
      \caption{(i) The graph $G$ with a perfect matching shown by dashed edges. (ii) An
        intermediate matching $M_2$. (iii) The matching $M_3$ after the
        next request arrives at $d$.}\label{fig:SAP}
    \end{center}
  \end{figure}


\begin{theorem}
\label{thm:SAP}
  The shortest augmenting path algorithm incurs in total  $O(n \log n)$
  augmentation cost  in expectation in the RO model.
\end{theorem}

\begin{proof}
  Fix an optimal matching $M^*$, and consider some time step during the
  algorithm's execution.  Suppose the current maximum matching $M$ has
  size $n-k$. As a thought experiment, if all the remaining $k$ vertices
  in $U$ are revealed at once, the symmetric difference $M^* \triangle
  M$ forms $k$ node-disjoint alternating paths from these remaining $k$
  vertices to unmatched nodes in $V$.  Augmenting along these paths
  would gives us the optimal matching. The sum of lengths of these paths
  is at most $|M^*|+|M| \leq 2n$. (Observe that the cost of the optimal
  solution on the remaining requests does increase over time, but only
  by a constant factor.) Now, since the next vertex is chosen uniformly at
  random, its augmenting path in the above collection---and hence the
  \emph{shortest} augmenting path from this vertex---has expected length
  at most $2n/k$. Now summing over all $k$ from $n$ down to
  $1$ gives a total expected cost of $2n \big(\frac1n + \frac{1}{n-1}
  + \ldots + \frac12 + 1 \big) = 2nH_n = O(n \log n)$, hence 
  proving the theorem.
\end{proof}



This $O(\log n)$-competitiveness guarantee happens to be tight for this
matching problem in the RO model; see Exercise~\ref{exer:minAug}.

\subsection{Bin Packing}
\label{sec:binpacking}

In the classical bin-packing problem (which you may recall from Chapter~8), the request sequence consists of
items of sizes $s_1, s_2, \ldots, s_n$; these items need to be packed into
bins of unit capacity. (We assume each $s_i \leq 1$.) The goal is to
minimize the number of bins used.  One popular algorithm is \textsc{Best
  Fit}:
  
\begin{algo}
  When the next item (with size $s_t$) arrives:
  \begin{enumerate}
  \item If the item does not fit in any currently used bin, put it in a
    new bin. Else,
  \item Put into a bin where the resulting empty space is minimized (i.e.,
    where it fits ``best'').
  \end{enumerate}
\end{algo}


\textsc{Best Fit} uses no more than $2\,OPT$ bins in the worst
case. Indeed, the sum of item sizes in any two of the bins is strictly
more than $1$, else we would never have started the later of these
bins. Hence we use at most $\lceil 2 \sum_i s_i \rceil$ bins, whereas $OPT$ must be at
least $\lceil \sum_i s_i \rceil$, since each bin holds at most unit total size.  A more sophisticated analysis
shows that \textsc{Best Fit} uses $1.7\, OPT + O(1)$ bins on any request
sequence, and this multiplicative factor of $1.7$ is the best possible.

The examples showing the lower bound of $1.7$ are intricate, but a
lower bound of $\nicefrac32$ is much easier to show, and also
illustrates why \textsc{Best Fit} does better in the RO model. Consider
$n/2$ items of size $\nicefrac12^- := \nicefrac12 - \eps$, followed by
$n/2$ items of size $\nicefrac12^+ := \nicefrac12 + \eps$. 
While the optimal solution uses $n/2$
bins, \textsc{Best Fit} uses $3n/4$ bins on this adversarial sequence,
since it pairs up the $\nicefrac12^-$ items with each other, and then
has to use one bin per $\nicefrac12^+$ item. On the other hand, in the
RO case, the imbalance between the two kinds of items behaves very
similar to a symmetric random walk on the integers. (It is
exactly such a random walk, but conditioned on starting and ending at
the origin). The number of $\nicefrac12^+$ items that occupy a bin by
themselves can now be bounded in terms of the maximum deviation from the
origin (see Exercise~\ref{ex:bin-p}), which is $O(\sqrt{n \log n}) =
o(OPT)$ with high probability. Hence on this instance \textsc{Best Fit} uses only
$(1+o(1))\, OPT$ bins in the RO model, compared to $1.5\, OPT$ in the
adversarial order. On general instances, we get:


\begin{theorem}
\label{thm:BP}
  The \textsc{Best Fit} algorithm uses at most $(1.5 + o(1))\, OPT$ bins
  in the RO setting.
\end{theorem}

\onlyforshort{The key insight in the proof of this result is a
  ``scaling'' result, saying that any $\eps n$-length subsequence of the
  input has an optimal value $\eps OPT$, plus lower-order terms. The
  proof uses the random order property and
  concentration-of-measure. Observe that the worst-case example does not
  satisfy this scaling property: the second half of that instance has
  optimal value $n/2$, the same as for the entire instance.  (Such a
  scaling property is often the crucial difference between the
  worst-case and RO moels: e.g., we used this in the algorithm for
  packing LPs in \S\ref{sec:LPsRandOrder}.)

 }

\saveforlong{
\begin{proof}(Sketch) 
The proof of this result has two insights. The first is a ``scaling'' result
which uses the random order property and concentration-of-measure.  It
says that for any $t$, if we look at the sub-instance given by the first
(or the last) $t$ requests, the optimal bin-packing solution on that
sub-instance uses about $\frac{t}{n} OPT$ bins, plus or minus
$O(\sqrt{OPT \log OPT})$ with high probability. Observe such a claim is
not true for adversarial instances (e.g., the one with $\nicefrac12^-$
and $\nicefrac12^+$-sized items).

Hence, for any choice of $t^*$ if we consider the two subinstances given
by requests in $\{1, \ldots, t^*\}$ and $\{t^*+1, \ldots, n\}$, the
optimum values on this prefix and suffix sum to $OPT$ plus lower order
terms, by the scaling result above. If we choose a time $t^*$ such that
the algorithm is $1.5$-competitive on both the subinstances, we'd be
done. 

The second insight is this choice of $t^*$. Indeed, let $t^*$ be
the last timestep when (a)~the arriving item of size $s_{t^*}$ satisfies
$s_{t^*} \leq \nicefrac13$, and also (b)~the algorithm places this item
into a bin with at least $\nicefrac12$ empty space before getting this
item. If there is no such time, set $t^* = 0$.
By the properties of \textsc{Best Fit}, every other bin (except the one
getting item $t^*$) must be full to at least
$1 - s_{t^*} \geq \nicefrac23$---so the algorithm is $1.5$-competitive
on the subinstance $\{1,\ldots,t^*\}$. Moreover, any bins opened after
this time (except maybe the last one) will contain two items of size
$> \nicefrac13$ or one item of size $> \nicefrac12$, and it is easy to
argue now that \textsc{Best Fit} is $1.5$-competitive for these bins.
\end{proof}
}

The exact performance of \textsc{Best Fit} in RO model remains
unresolved: the best known lower bound  uses
$1.07\, OPT$ bins. Can we close this gap? Also, can we analyze other
common heuristics in this model? E.g., \textsc{First Fit} places the
next request in the bin that was started the earliest, and can
accommodate the item. Exercise~\ref{exer:binPacking} asks you to
show that the \textsc{Next Fit} heuristic does not
benefit from the random ordering, and has a competitive ratio of $2$ in
both the adversarial and RO models.


\subsection{Facility Location}
\label{sec:facility-location}

A slightly different algorithmic intuition is used for the \emph{online
  facility location} problem, which is related to the $k$-means and
$k$-median clustering problems. 
In this problem, we are given a metric space $(V,d)$ with point set $V$, and distances
$d: V\times V \to \mathbb{R}_{\geq 0}$ satisfying the triangle
inequality.  Let $f \geq 0$ be the cost of opening a facility; the
algorithm can be extended to cases where different locations have
different facility costs. Each
request is specified by a point in the metric space, and let
$R_t = \{r_1, \ldots, r_t\}$ be the (multi)-set of request points that
arrive by time $t$. A solution at time $t$ is a set $F_t \sse V$ of
``facilities'' whose opening cost is $f\cdot |F_t|$, and whose
connection cost is the sum of distances from every request to its
closest facility in $F_t$, i.e.,
$\sum_{j \in R_t} \min_{i \in F_t} d(j,i)$. An open facility remains
open forever, so we require $F_{t-1} \sse F_t$. We want the algorithm's
total cost (i.e., the opening plus connection costs) at time $t$ to be
at most a constant times the optimal total cost for $R_t$ in the RO
model. Such a result is impossible in the adversarial arrival model,
where a tight $\Theta(\frac{\log n}{\log \log n})$ worst-case
competitiveness is known.

There is a tension between the two components of the cost: opening more
facilities increases the opening cost, but reduces the connection cost. 
Also, when request $r_t$ arrives,
if its distance to its closest facility in $F_{t-1}$ is more than $f$,
it is definitely better (in a greedy sense) to open a new facility at
$r_t$ and pay the opening cost of $f$, than to pay the connection cost
more than $f$. This suggests the following algorithm:

\begin{algo}
  When a request $r_t$ arrives:
  \begin{enumerate}
  \item Let $d_t := \min_{i \in F_{t-1}} d(r_t, i)$ be its distance to
    the closest facility in $F_{t-1}$. 
  \item Set $F_t \gets F_{t-1} \cup \{r_t\}$ with probability
    $p_t := \min\{1, d_t/f\}$, and $F_t \gets F_{t-1}$ otherwise.
  \end{enumerate}
\end{algo}

Observe that the choice of $p_t$ approximately balances the expected
opening cost $p_t \cdot f \leq d_t$ with the expected connection cost
$(1-p_t) d_t \leq d_t$. Moreover, since the set of facilities increases
over time, a request may be reassigned to a closer facility later in the
algorithm; however, the analysis works even assuming the request $r_t$
is permanently assigned to its closest facility in $F_t$.

\begin{theorem}
  \label{thm:FACLOC}
  The above algorithm is $O(1)$-competitive in the RO model. 
\end{theorem}

The insight behind the proof is a charging argument that first
classifies each request as ``easy'' (if they are close to a facility in
the optimal solution, and hence cheap) or ``difficult'' (if they are far
from their facility). There are an equal number of each type, and the
random permutation ensures that easy and difficult requests are roughly
interleaved. This way, each difficult request can be paired with its
preceding easy one, and this pairing can be used to bound their cost.

\saveforlong{
\begin{proof}(Sketch)
  For simplicity assume that $f = 1$, and all distances are at most $1$,
  so that $\min \{1, d_t/f\} = d_t$.  Consider some facility $i^*$ in
  the optimal solution, and let $S$ be the set of request indices served
  by it, i.e., those for which the closest facility is $i^*$.  Let
  $d^* := \frac{1}{|S|} \sum_{j \in S} d(r_j,i^*)$ be the distance from
  an average request point in $S$ to the facility~$i^*$. The optimal
  cost to open facility $i^*$ and serve these requests is
  \[ OPT_S := \textstyle f + \sum_{j \in S} d(r_j,i^*) = 1 + |S| d^*. \]
  Observe that in our algorithm, each request point $r_t$ pays $f=1$
  with probability $\min(1,d_t/f) = d_t$, and $d_t$ otherwise. Summing
  these two, the expected cost for request $r_t$ is at most $2d_t$, and
  hence it suffices to show that the algorithm's cost for requests in $S$, 
  which is $\sum_{j \in S} d_j$, is at most $O(1) \cdot OPT_S$.

  Sort all requests in $S$ in order of their distance from the facility
  $i^*$, and call the first half the ``close'' requests $C$, and the
  second half the ``distant'' requests $D$.  The close requests are at
  distance at most $2d^*$ from the facility by Markov's
  inequality. Note that once we open a facility $j$ in $C$, each of the
  subsequent requests $j'$ pays at most
  $d_{j'} \leq d(j',j) \leq d(j',i^*) + d(j,i^*) \leq d(j',i^*) + 2d^*$,
  which sum to at most $OPT_S + 2d^*|S| \leq 3OPT_S$. Hence, we need to
  bound the cost incurred by requests that arrive before any facility is
  opened inside $S$.

  We first ``charge'' the distant requests to the close requests.  Indeed,
  consider the (random) arrival ordering: each distant request $j \in D$
  charges to its preceding close request $j' \in C$ in this ordering (if
  any). Since $j$ can use the same facility as $j'$, the distance
  $d_j \leq d_{j'} + d(r_j, r_{j'}) \leq d_{j'} + d(r_j, i^*) +
  d(r_{j'}, i^*) $. Now by the random ordering and the fact that
  $|C| = |D|$, each close request $j' \in C$ only has a single distant
  request $j \in D$ charging to it in expectation, so summing gives that
  $\E[\sum_{j \in D} d_j] \leq \E[\sum_{j \in C} d_j] + O(OPT_S)$.  (The
  distant requests that come before all close requests have to be
  charged separately as below; we skip the details.)
  
  Finally, it remains to bound the expected cost for the close requests
  before we open any close facilities: for each request $j$, we pay
  $d_j$ and the process stops (since we open a facility) with
  probability $d_j$, else we continue. A simple calculation shows that
  the total expected cost is at most $1$.
\end{proof}

To summarize: we used the random ordering to create a matching between
the distant (``difficult'') and the close (``easy'') requests, and charged
the former to the latter. The latter we bounded by the triangle
inequality, and the fact that we open facilities with probability
proportional to how much we pay at each step.
}








\newcommand{\calD}{\mathcal{D}}
\section{Related Models and Extensions}
\label{sec:others}

There have been  other online models related to RO arrival.
Broadly, these models can be classified either as ``adding more
randomness'' to the  RO model by making further stochastic assumptions on
the arrivals, or as ``limiting randomness'' where the arrival sequence
need not be uniformly random. The former lets us exploit the
increased stochasticity to design algorithms with better performance
guarantees; the latter help us quantify the robustness of the
algorithms, and the limitations of
the RO model. 


\subsection{Adding More Randomness} \label{sec:addRandom}

The RO model is at least as general as the
\emph{i.i.d.\ model}, which assumes a probability distribution $\calD$
over potential requests where each request is an independent draw from
the distribution $\calD$. Hence, all the above RO results immediately
translate to the i.i.d.\ model. Is the converse also true---i.e., can we
obtain identical algorithmic results in the two models? The next case
study answers this question negatively, and then illustrates how to use
the knowledge of the underlying distribution to perform better.

\subsubsection{Steiner Tree in the RO model}
\label{sec:steiner-tree}

In the online Steiner tree problem, we are given a metric space $(V,d)$, which can be
thought of as a complete graph with edge weights $d(u,v)$; each request
is a vertex in $V$. Let $R_t = \{r_1, \ldots, r_t\}$ be the
(multi)-set of request vertices that arrive by time~$t$. When a 
request $r_t$ arrives, the algorithm must pick some edges $E_t \subseteq
\binom{V}{2}$, so that edges $E_1 \cup \cdots \cup E_t$ picked till
now connect $R_t$ into a single component. The cost of each edge
$\{u,v\}$ is its length $d(u,v)$, and the goal is to minimize the total cost.

The first algorithm we may imagine is the greedy algorithm, which picks a
single edge connecting $r_t$ to the closest previous request in
$R_{t-1}$. This greedy algorithm is $O(\log n)$-competitive for any
request sequence of length $n$ in the worst-case. Surprisingly, there
are lower bounds of $\Omega(\log n)$, not just for adversarial arrivals,
but also for the RO model. Let us discuss how the logarithmic
lower bound for the adversarial model translates to one for 
the RO model.

There are two properties of the Steiner tree problem that make this
transformation possible. The first is that duplicating requests does not
change the cost of the Steiner tree, but making many copies of a request
makes it likely that one of these copies will appear early in the RO
sequence. Hence, if we take a fixed request sequence $\sigma$, duplicate
the $i^{th}$ request $C^{n-i}$ times (for some large $C > 1$), apply a uniform random
permutation, and remove all but the first copy of each original request,
the result looks close to $\sigma$ with high probability. 
Of course, the sequence length increases from $n$
to $\approx C^n$, and hence the lower bound goes from 
being logarithmic to being doubly
logarithmic in the sequence length.

We now use a second property of Steiner tree: the worst-case examples
consist of $n$ requests that can be given in $\log n$ batches, with
the $i^{th}$ batch containing $\approx 2^i$ requests---it turns out that
giving so much information in parallel does not help the algorithm. Since we
don't care about the relative ordering within the batches, we can
duplicate the requests in the $i^{th}$ \emph{batch} $C^i$
times, thereby making the resulting request sequences of length $\leq
C^{1+\log n}$. It is now easy to set $C = n$ to get a
lower bound of $\Omega(\frac{\log n}{\log \log n})$, but a careful
 analysis allows us to set $C$ to be a constant, and get an
$\Omega(\log n)$ lower bound for the RO setting.

\subsubsection{Steiner Tree in the i.i.d. model}
\label{sec:steiner-tree-iid}

Given this lower bound for the RO model, what if we make stronger
assumptions about the randomness in the input?  What if the arrivals are
i.i.d.\ draws from a probability distribution?  We now have to make an
important distinction, whether the distribution is known to the
algorithm or not. The lower bound of the previous section can easily be
extended to the case where arrivals are from an \emph{unknown}
distribution, so our only hope for a positive result is to
consider the i.i.d.\ model with \emph{known} distributions. In other
words, 
each  request is a random vertex of the graph, where
vertex $v$ is requested with a known probability $p_v \geq 0$ (and
$\sum_v p_v = 1$). Let $\mathbf{p} = (p_1, p_2, \ldots, p_{|V|})$ be the
vector of these probabilities. For simplicity, assume  that we know the
length $n$ of the request sequence.
The \emph{augmented greedy} algorithm is the following:

\begin{algo}
  \begin{enumerate}
  \item Let $A$ be the (multi-)set of $n-1$ i.i.d.\ samples from the
     distribution $\bf{p}$, plus the first request $r_1$. 
   \item Build a minimum spanning tree $T$ connecting all the vertices
     in $A$.
     
   \item For each subsequent request $r_i$ (for $i \geq 2$):
connect $r_i$ to the closest vertex in $A \cup R_{i-1}$ using a
       direct edge.
  \end{enumerate}
\end{algo}
Note that our algorithm requires minimal knowledge of the underlying
distribution: it merely takes a set of samples that are
stochastically identical to the actual request sequence, and builds the
``anticipatory'' tree connecting this sample. Now the hope is that the
real requests will look similar to the samples, and hence will have
close-by vertices to which they can connect.

\begin{theorem}
\label{thm:Steiner-iid}
  The augmented greedy algorithm is $4$-competitive for the Steiner tree problem in
  the setting of i.i.d.\ requests with known distributions.
\end{theorem}

\begin{proof}
  Since set $A$ is drawn from the same distribution as the actual
  request sequence $R_n$, the expected optimal Steiner tree on $A$ also  
  costs $OPT$. The minimum spanning tree  to connect up $A$  
  is known to give a
  $2$-approximate Steiner tree
  (Exercise~\ref{exer:steinerProphet}), so the expected cost
  for $T$ is $2OPT$. 

  Next, we need to bound the expected cost of connecting $r_t$ to the
  previous tree for $t \geq 2$. Let the samples in $A$ be called $a_2,
  a_3, \ldots, a_n$. Root the tree $T$ at $r_1$, and let the ``share''
  of $a_t$ from $T$ be the cost of the first edge on the path from $a_t$
  to the root. The sum of shares equals the cost of $T$.  Now, the cost
  to connect $r_t$ is at most the expected minimum distance from $r_t$
  to a vertex in $A \setminus \{a_t\}$. But $r_t$ and $a_t$ are from the
  same distribution, so this expected minimum distance is bounded by the
  expected distance from $a_t$ to its closest neighbor in $T$, which is
  at most the expected share of $a_t$. Summing, the connection costs for
  the $r_t$ requests is at most the expected cost of $T$, i.e., at most
  $2OPT$. This completes the proof.
\end{proof}

The proof above extends to the setting where different requests are
drawn from different known distributions; see Exercise~\ref{exer:steinerProphet}.

\subsection{Reducing the Randomness}
\label{sec:redRandom}

Do we need the order of items to be uniformly random, or can weaker
assumptions suffice for the problems we care about? This question was
partially addressed in \S\ref{sec:order-obliv-algor} where we saw
order-oblivious algorithms. Recall: these algorithms assume a
less-demanding arrival model, where a random fraction of the adversarial
input set $S$ is revealed to the algorithm in a first phase, and the
remaining input set arrives in an adversarial order in the second
phase. We now discuss some other models that have been proposed to
reduce the amount of randomness required from the input. While some
remarkable results have been obtained in these directions, there is
still much
to  explore.

\subsubsection{Entropy of the random arrival order}
\label{sec:entr-rand-arriv}

One principled way of quantifying the randomness is to measure the
entropy of the input sequences: a uniformly random permutation on $n$
items has entropy $\log (n!) = O(n \log n)$, whereas order-oblivious
algorithms (where each item is put randomly in either phase) require at most $\log {{n}\choose{n/2}} \leq n$ bits of entropy. Are
there arrival order distributions with even less entropy for which we
can give good algorithms?
 
This line of research was initiated by \cite{KKN}, who showed the
existence of  
arrival-order distributions with entropy only $O(\log\log n)$ that allow
$e$-competitive algorithms for the single-item secretary problem
(and also for some simple multiple-item problems). Moreover, they showed
tightness of their results---for any arrival distribution with
$o(\log\log n)$ entropy 
no online algorithm can be $O(1)$-competitive.  This work also 
defines a notion of
``almost $k$-wise uniformity'', which requires that the induced
distribution on every subset of $k$ items be close to uniform. They
show that this property and its variants suffice for some of
the  algorithms, but not for all.

A different perspective is the following: since the performance analysis
for an algorithm in the RO model depends only on certain randomness
properties of the input sequence (which are implied by the random
ordering), it may be meaningful in some cases to (empirically) verify
these specific properties on the actual input stream. E.g., \cite{BCG}
used this approach to explain the experimental efficacy of
their algorithm computing personalized pageranks in the RO model.
 
\subsubsection{Robustness and the RO Model}
\label{sec:robust-algorithms}

The RO model assumes that the adversary first chooses all the item
values, and then the arrival order is perturbed at random according to
some specified process. While this is a very appealing framework, one
concern is that the algorithms may over-fit the model. What if, as in
some other semi-random models, the adversary gets to make a small number of
changes \emph{after} the randomness has been added? Alternatively, what
if some parts of the input must remain in adversarial order, and the
remainder is randomly permuted?  For instance, say the adversary is
allowed to specify a single item which must arrive at some specific
position in the input sequence, or it is allowed to change the position
of a single item after the randomness has been added. Most current
algorithms fail when faced with such modest changes to the model. E.g.,
the 37\%-algorithm picks nothing if the adversary presents a large item
in the beginning. Of course, these algorithms were not designed to be
robust: but can we get analogous results  even if
the input sequence is slightly corrupted?


One approach is to give ``best-of-both-worlds'' algorithms that achieve
a good performance when the input is randomly permuted, and which also
have a good worst-case performance in all cases. For instance,
\cite{MirrokniOZ} and \cite{Raghvendra} give such results for online ad
allocation and min-cost matching, respectively. Since the secretary
problem has poor performance in the worst case, we may want
more refined guarantees, that the performance
degrades with the amount of corruption. Here is a different semi-random
model for the multiple-secretary problem from \S\ref{sec:k-items}. In
the \emph{Byzantine} model, the adversary not only chooses the values of
all $n$ items, it also chooses the relative or absolute
order of some $\eps n$ of these items. The remaining $(1-\eps)n$
``good'' items are then randomly permuted within the remaining
positions.  The goal is now to compare to the top $k$ items among only
these good items. Some preliminary results are known for this
model~\citep{BradacGSZ-arXiv19}, but many questions remain open. In
general, getting robust algorithms for secretary problems, or other
optimization problems considered in this chapter, remains an important
direction to explore.

\subsection{Extending Random-Order Algorithms to Other Models}

Algorithms for the RO model can help in designing good algorithms for
similar models. One such example is for the \emph{prophet} model, which
is closely
related to the optimal stopping problem in \S1.1.4 of
Chapter~8. In this model, we are given $n$ independent prize-value distributions $D_1,
\ldots, D_n$, and then presented with draws $v_i \sim D_i$ from these
distributions in \emph{adversarial} order. The threshold rule from
Chapter~8, which picks the first prize with value above some threshold
computable from just the distributions, gets expected value at least
$\frac12 \E[\max_i v_i]$. Note the differences between the prophet and
RO models: the prophet model assumes more about
the values---namely, that they are drawn from
the given distributions---but less about the ordering, since the items
can be presented in an adversarial order. Interestingly,
order-oblivious algorithms in the RO model can be used to get algorithms
in the prophet model.

Indeed, suppose we only have limited access to the distributions $D_i$
in the prophet model: we can only get information about them by drawing
a few samples from each distribution.  (Clearly we need at least one
sample from the distributions, else we would be back in the online
adversarial model.) Can we use these samples to get algorithms in this
limited-access prophet model for some packing constraint family $\calF$?
The next theorem shows 
we can convert order-oblivious algorithms for the RO model
to this setting using only a single
sample from each distribution.
 
\begin{theorem} \label{thm:orderOblivToProph} Given an
  $\alpha$-competitive order-oblivious online algorithm for a packing
  problem $\calF$, there exists an $\alpha$-competitive algorithm for
  the corresponding prophet model with unknown probability
  distributions, assuming we have access to one (independent) sample
  from each distribution.
\end{theorem}

The  idea is to choose a random subset of items to be
presented to the order-oblivious algorithm in the first phase; for these
items we send in the sampled values available to us, and for the
remaining items we use their values among the actual arrivals. The details
are left as Exercise~\ref{exer:orderOblivToProph}.

\section{Notes}
\label{sec:notes}

The classical secretary problem and its variants have long been studied
in optimal stopping theory; see~\cite{Ferguson-Journal89} for a
historical survey. In computer science, the RO model has been used, e.g., for computational geometry problems,
to get fast and elegant 
algorithms for problems like convex hulls and linear
programming; see \cite{Seidel} for a survey in the
context of the \emph{backwards analysis} technique.
%
The secretary problem 
has gained broader attention 
due to connections to strategyproof
mechanism design for online auctions 
\citep{HajiaghayiKP04,kleinberg}.  
Theorem~\ref{thm:secyTight} is due to~\cite{GilbertMosteller}.

\S\ref{sec:order-obliv-algor}: The notion of an order-oblivious
algorithm was first defined by \cite{AKW}. The order-oblivious
multiple-secretary algorithm is folklore. The matroid secretary
problem was proposed by~\cite{BIK07}; 
Theorem~\ref{thm:maxWtdForestLog}
is an adaptation of their
$O(\log r)$-competitive algorithm for general matroids. Theorem~\ref{thm:maxWtdForestO1},
due to \cite{KorulaPal-ICALP09}, extends to a
${2e}$-competitive order-adaptive algorithm. The current best algorithm is
$4$-competitive~\citep{SotoTV-SODA18}. The only lower bound known is the 
factor of $e$ from Theorems~\ref{thm:secy} and~\ref{thm:secyTight}, even for arbitrary matroids, whereas the best 
algorithm for general matroids has competitive ratio $O(\log \log
\text{rank})$~\citep{Lachish14,FeldmanSZ15}. See the survey
by~\cite{Dinitz13} for work leading up to these results.

\S\ref{sec:orderAdap}: The order-adaptive algorithms for
multiple-secretary (Theorem~\ref{thm:kItemsAdap}) and for packing LPs
(Theorem~\ref{thm:LPs}) are based on the work of \cite{AWY14}. The
former result can be improved to give 
$(1 - O(\sqrt{\nicefrac{1}{k}}))$-competitiveness 
\citep{kleinberg}.
Extending work on the AdWords problem (see the monograph
by~\cite{Mehta}), \cite{DevanurHayes09} studied packing
LPs in the RO model. 
The optimal results have 
$(1 - O(\sqrt{(\log d_{\text{nnz}})/k}))$-competitiveness \citep{KRTV13},
where $d_{\text{nnz}}$ is the maximum number of non-zeros in any column;
these are based on the solve-ignoring-past-decisions approach we used
for max-value matchings.
\cite{Rubinstein-STOC16} and \cite{RS-SODA17} gave 
$O(\poly\log n)$-competitive algorithms for general packing problems
for subadditive functions.
Theorem~\ref{thm:MWM-basic} 
(and extensions to combinatorial auctions)  
is due to \cite{KesselheimRTV-ESA13}.

\S\ref{sec:comb}: Theorem~\ref{thm:SAP} about shortest augmenting paths is due to
\cite{ChaudhuriDKL-INFOCOM09}. A worst-case result of $O(\log^2 n)$ was
given by \cite{BernsteinHR18}; closing this gap remains an open problem.
The analysis of \textsc{Best Fit} in Theorem~\ref{thm:BP}
is by \cite{kenyon}. Theorem~\ref{thm:FACLOC} for facility
 location is by \cite{Meyerson}; see \cite{MeyersonMP01} for
other network design problems in the RO setting. The tight
nearly-logarithmic competitiveness for adversarial arrivals is due to
\cite{Fotakis08}. 


\S\ref{sec:addRandom}:
Theorem~\ref{thm:Steiner-iid} 
for Steiner tree 
is by \cite{GargGLS-SODA08}. \cite{GrandoniGLMSS13} and \cite{DehghaniEHLS17} give algorithms for set cover and $k$-server in the
i.i.d.\ or prophet models. The gap between 
the RO 
and i.i.d.\ models (with \emph{unknown} distributions) remains an interesting
direction to explore. \cite{CorreaDFS-EC19} show that the single-item
 problem has the same competitiveness in both models; can we
show similar results (or gaps) for other problems?


\S\ref{sec:redRandom}: \cite{KKN} study connections between 
entropy
and the RO model.
The RO model with corruptions was
proposed by \cite{BradacGSZ-arXiv19}, who also give
$(1-\eps)$-competitive algorithms for the multiple-secretary problem
with weak estimates on the optimal value.
A similar model for online
matching for mixed (stochastic and worst-case) arrivals was studied
by \cite{EKM-EC15}. Finally, Theorem~\ref{thm:orderOblivToProph} to
design prophet inequalities from samples is by \cite{AKW}. Connections
between the RO model and prophet inequalities have also been studied via 
 the prophet secretary model where both the value distributions are given and
 the arrivals are in a uniformly random-order~\citep{EHLM-SIDMA17,EHKS-SODA18}. 



  \subsection*{Acknowledgments}
We thank  Tim Roughgarden, C. Seshadhri, Matt Weinberg, and Uri Feige
 for their comments on an initial draft of this chapter.


{
\bibliography{chap11}
\bibliographystyle{cambridgeauthordate}
}

\section*{Exercises}

\begin{enumerate}
\item \label{exer:multItemConstant} Show that both algorithms proposed
  above Theorem~\ref{thm:kItemObiv} for the multiple secretary problem
  achieve expected value $\Omega(V^\star)$.

\item \label{exer:packLowerBound} Show why for a general packing
  constraint family $\calF$, i.e. $A \in \calF$ and $B \subseteq A$
  implies $B \in \calF$, no online algorithm has
  $o(\frac{\log n}{\log\log n})$-competitiveness. 
 [Hint: Imagine $n$
  elements in a $\sqrt{n}\times \sqrt{n}$ matrix and $\calF$
  consists of subsets of columns.]
  
\item  \label{exer:minAug}
  Consider a cycle on $2n$ vertices, and hence it has a perfect matching.
  Show that the shortest augmenting path algorithm for minimizing 
  augmentations in online matching from
  \S\ref{sec:matching} has  $\Omega(n \log n)$ cost in expectation.

\item \label{ex:bin-p} Suppose $n/2$ items of size $\nicefrac12^- := \nicefrac12 - \e$
  and $n/2$ items of size $\nicefrac12^+ := \nicefrac12 + \e$ are presented in 
  a random
  order to the \textsc{Best Fit} bin packing heuristic from
  \S\ref{sec:binpacking}. Define the imbalance $I_t$ after $t$ items to
  be the number of $\nicefrac12^+$ items minus the number of
  $\nicefrac12^-$ items. Show that the number of bins which have only a
  single item (and hence waste about $\nicefrac12$ space) is at most
  $(\max_t I_t) - (\min_t I_t)$. Use a Chernoff-Hoeffding bound to prove
  this is at most $O(\sqrt{n \log n})$ with probability $1 - 1/\poly(n)$.

\item \label{exer:binPacking} In the \textsc{Next Fit} heuristic for the
  bin packing problem from \S\ref{sec:binpacking}, the next item is
  added to the current bin if it can accommodate the item, otherwise we
  put the item into a new bin. Show that this algorithm has a competitive
  ratio of $2$ in both the adversarial and RO models.

\item \label{exer:steinerProphet} Show that for a set of requests in a
  metric space, 
the minimum spanning tree on this subset gives a $2$-approximate
solution to the Steiner tree. Also, extend the Steiner tree algorithm from
  the i.i.d.\ model to the prophet model, where the $n$ requests are 
  drawn from $n$ independent (possibly different) known distributions 
  $\mathbf{p}^t$ over the
  vertices. 
  
\ignore{\item \label{exer:singleItemFixedThresh} Show that
  Theorem~\ref{thm:singleItemIID} can be extended to the prophet
  secretary model where arrivals are from known independent
  distributions $X_1, \ldots, X_n$ in a random order by setting a fixed
  threshold $\tau$ where $\Pr[\max_i X_i >\tau] = 1-1/e$. 
  [Hint: Use
  the inequality that $\forall t,x \in [0,1]$, we have
  $\log(1-tx) \geq t\cdot \log(1-x)$.] \alert{Drop?}}

\item \label{exer:orderOblivToProph}
Prove Theorem~\ref{thm:orderOblivToProph}. 

\item \label{exer:byzSecret}
Suppose the input consists of $(1-\e)n$ good items and  $\e n$ bad items.  The adversary decides all $n$ item values, and also the locations of bad items. The good items are then randomly permuted in the remaining positions. If $v^*$ denotes the value of the largest good item, show there is no online algorithm with expected value $\Omega(v^*/(\e n))$.
[Hint: There is only one non-zero good item and all the bad items have values much smaller than $v^*$. The bad items are arranged as the  single-item adversarial arrival lower bound instance.]

\end{enumerate}


\end{document}